\documentclass[journal,twoside]{IEEEtran}

\usepackage{array}
\usepackage[latin1]{inputenc}
\usepackage{amsmath}
\usepackage{amsfonts}
\usepackage{amssymb}
\usepackage{amsthm}
\usepackage[final]{graphicx}
\usepackage{algorithm,algpseudocode}
\usepackage{multirow}
\usepackage{cite}
\usepackage{subfig}

\newtheorem{theorem}{Theorem}

\newcommand{\mB}[1]{\ensuremath{m_{B,#1}}}
\newcommand{\mb}[1]{\ensuremath{m_{b,#1}}}
\newcommand{\thetav}{\ensuremath{\boldsymbol{\theta}}}
\newcommand{\thetaside}{\ensuremath{\hat{\thetav}_\mathsf{S}}}
\newcommand{\thetacentral}{\ensuremath{\hat{\thetav}_\mathsf{C}}}
\newcommand{\x}{\ensuremath{\mathbf{x}}}
\newcommand{\y}{\ensuremath{\mathbf{y}}}
\newcommand{\Ri}{\ensuremath{\mathbb{R}}}
\newcommand{\indfun}{\mathcal{I}}

\begin{document}

\title{Graded quantization for multiple description coding of compressive measurements}
\author{Diego Valsesia,~\IEEEmembership{Student Member,~IEEE,} Giulio Coluccia,~\IEEEmembership{Member,~IEEE,} and\\ Enrico Magli,~\IEEEmembership{Senior Member,~IEEE}
\thanks{The authors are with the Department of Electronics and Telecommunications of Politecnico di Torino. Their work has received funding from the European Research Council under the European Community's Seventh Framework Programme (FP7/2007-2013) / ERC Grant agreement number 279848.}}

\markboth{IEEE Transactions on Communications}%
{Submitted paper}

\maketitle

\begin{abstract}

Compressed sensing (CS) is an emerging paradigm for acquisition of compressed representations of a sparse signal. Its low complexity is appealing for resource-constrained scenarios like sensor networks. However, such scenarios are often coupled with unreliable communication channels and providing robust transmission of the acquired data to a receiver is an issue. Multiple description coding (MDC) effectively combats channel losses for systems without feedback, thus raising the interest in developing MDC methods explicitly designed for the CS framework, and exploiting its properties. We propose a method called Graded Quantization (CS-GQ) that leverages the democratic property of compressive measurements to effectively implement MDC, and we provide methods to optimize its performance. A novel decoding algorithm based on the alternating directions method of multipliers is derived to reconstruct signals from a limited number of received descriptions. Simulations are performed to assess the performance of CS-GQ against other methods in presence of packet losses. The proposed method is successful at providing robust coding of CS measurements and outperforms other schemes for the considered test metrics.    
\end{abstract}

\IEEEpeerreviewmaketitle

\section{Introduction}
Compressed sensing (CS) \cite{CS_donoho} is a novel theory for sampling and acquisition of sparse and compressible signals. The traditional paradigm based on sampling a signal according to the Nyquist/Shannon theorem followed by compression can be replaced by the acquisition of a small number of linear measurements, in the form of random projections. This is very appealing for low-complexity systems, such as low-power sensor motes, where classic acquisition followed by compression could be expensive in terms of energy consumption and computational demands. Such systems typically need to transmit the acquired data to a receiver over unreliable channels, thus raising the issue of robustness of the transmission against channel losses and at the same time imposing a constraint on the complexity of the adopted methods. The framework of multiple description coding (MDC) allows to increase the robustness by creating multiple correlated representations of the data to be transmitted. The quality of the decoded data will then depend on the number of descriptions available at the receiver, where \emph{side decoders} can recover a lower quality version of the data from few received descriptions, whereas the \emph{central decoder} can achieve the maximum quality when all the descriptions are received. A few works studied the problem of MDC in the CS framework. An early work on the topic \cite{CSApproachFrameMDC} borrows from the theory of sparse representations to derive a two-description coding scheme through a frame synthesis operator. However, the work does not consider random projections. The work in \cite{CSMDC} presents a method to generate descriptions of an image by sensing each of them separately, \emph{i.e.}, preprocessing the image to split it into two sub-images and then sense the wavelet coefficients of each separately. We will not consider this kind of approach in the paper because we want to explicitly avoid any preprocessing of the signal before acquisition. Indeed, compressive measurements could be directly obtained by specialized hardware (\emph{e.g.}, \cite{SPCamera, chan2008terahertz, tomasparallel}), thus hindering any preprocessing of the data. Deng \emph{et al.}, \cite{RobustImageCodingCS} argue that the democratic property of random projections makes compressive sensing image coding robust to channel losses. They just partition the measurements into packets so that the quality of the decoded image will depend on the number of received packets. However, they do not consider a practical packetization problem, as they employ very small packets containing few measurements for each transmission. Such system is actually not sensible due to the high packetization overhead. 

This paper builds on the MDC mechanism called graded quantization (CS-GQ), originally proposed in \cite{ValsesiaICASSP},  exploiting the democratic property of compressive measurements. The principle behind CS-GQ is that multiple copies of the measurement vector can be used as descriptions. Each description is then partitioned into sets of samples, and each set is quantized with different quality. The principle is similar to \cite{subbalakshmi2002multiple, jiang1998multiple, mdsq, samarawickrama2010channel, tillo2010multiple, dumitrescu2010unequal, akyol2007flexible, baccaglini2007flexible, tillo2007multiple, bajic2003domain, mdsq2011}, but has never been applied in conjunction with CS. 

In this paper, we give several novel contributions with respect to \cite{ValsesiaICASSP}. First, we discuss how methods based on segmentation of measurements such as \cite{RobustImageCodingCS} are special cases of CS-GQ and, contrary to the present literature on MDC for CS, we carefully consider packetization issues. The use of CS for low-complexity and low-energy sensor networks motivates us to pay particular attention to the use of the proposed method with common data link layer protocols such as IEEE 802.15.4.  Moreover, we propose techniques to improve the performance of CS-GQ, to optimize its parameters, and we introduce a novel algorithm based on the alternating directions method of multipliers (ADMM) \cite{boydadmm} to implement the side decoder. On the theoretical side, perfomance analysis is conducted by providing reconstruction guarantees for the side decoder, using both the actual decoder and an ideal version based on the oracle decoder, which serves as a limit case for performance assessment. On the experimental side, simulations are performed for various scenarios. The paper is organized as follows: Sec. \ref{sec:background} explains some introductory background concepts, Sec. \ref{cs_gq} introduces CS-GQ and Sec. \ref{admm} the side decoding algorithm based on ADMM. Sec. \ref{sec:theory} presents theorems guaranteeing stable side decoding and discusses the performance of oracle-based side and central decoders. Sec. \ref{sec:rd_opt} deals with the optimization of the CS-GQ parameters as function of the channel loss probability and discusses packetization issues. Finally, Sec. \ref{numerical} reports the results of simulations comparing various MDC methods, the performance of CS-GQ in presence of channel losses using two different channel models and shows that the robustness of multiple descriptions can improve scene recognition accuracy in a sensor network application.

\section{Background}
\label{sec:background}

\subsection{Compressed sensing}

In the standard CS framework, introduced in \cite{CS_donoho,candes2006nos}, a signal $\x\in\Ri^{n\times 1}$ which has  a sparse representation in some basis $\Psi\in\Ri^{n\times n}$, \textit{i.e.} $\x = \Psi \thetav,\quad \Vert\thetav\Vert_0 = k,\quad k\ll n $ can be recovered from a smaller vector of noisy linear measurements $\y = \Phi\x + \boldsymbol{e}$, $\y\in\Ri^{m\times 1}$ and $k<m<n$,  where $\Phi\in\Ri^{m\times n}$ is the \emph{sensing matrix} and $\boldsymbol{e}\in\Ri^{m\times 1}$ is the vector representing additive noise such that $\Vert\boldsymbol{e}\Vert_2 < \varepsilon$, using \cite{bpdn}

\begin{equation}\label{eq:CS_recovery_relaxed}
\hat{\thetav}=\arg\min_{\thetav}\Vert\thetav\Vert_1\ \quad \text{s.t.}\quad \Vert\Phi\Psi\thetav - \y\Vert_2 < \varepsilon~
\end{equation}
and $\widehat{\x}=\Psi\widehat{\boldsymbol{\theta}}$, known as Basis Pursuit DeNoising, provided that $m = O(k\log(n/k))$ \cite{candes2006nos} and that each submatrix consisting of $k$ columns of $\Phi\Psi$  is (almost) distance preserving \cite[Definition 1.3]{davenport2012introduction}. The latter condition is the \emph{Restricted Isometry Property} (RIP). Formally, the matrix $\Phi\Psi$ satisfies the RIP of order $k$ if $\exists \delta_k \in (0,1]$ such that, for any $\thetav$ with $\Vert\thetav\Vert_0 \le k$:
$$
(1-\delta_k)\Vert\thetav\Vert_2^2\le\Vert\Phi\Psi\thetav\Vert_2^2\le(1+\delta_k)\Vert\thetav\Vert_2^2,
$$
where $\delta_k$ is the RIP constant of order $k$. It has been shown in \cite{baraniuk2008spr} that when $\Phi$ is an i.i.d. random matrix drawn from any subgaussian distribution and $\Psi$ is an orthogonal matrix, $\Phi\Psi$ satisfies the RIP with overwhelming probability. Moreover, using a random matrix as sensing matrix, ensures the \emph{democratic property} of compressive measurements $\y$ \cite{Democracy_Laska}. Indeed, it can be seen that each measurement has roughly the same importance as the others, in the sense that none of them improves or degrades the quality of the reconstruction significantly more than the others. This property is the key ingredient of our proposed CS-GQ technique, as described in Section~\ref{cs_gq}.

\subsection{Multiple description coding}
\label{mdc}
MDC \cite{mdcMagazine} is a way of coding an information source that is resilient to packet losses. The multiple description technique allows to create multiple correlated representations of the original information source, each carrying enough information to decode the source with a certain fidelity. Losing a description will not make the received data unusable since each description can be decoded separately, albeit providing a limited quality. However, the best decoding quality is obtained when all the descriptions are available. In the framework of MDC, one can identify the so-called \emph{central} decoder and the \emph{side} decoders. The role of the central decoder is to decode the source when all descriptions have been received, thus achieving the so-called central distortion. The side decoder recovers the source with a lower fidelity (side distortion) since only a subset of the descriptions has been received. This scheme is depicted in Fig.~\ref{fig:mdc} for the simple two-description case. One may seek to create balanced or unbalanced descriptions depending on their individual contribution to the final quality of the recovered source. Several approaches to MDC have been studied in the literature. Among others, it is possible to identify approaches based on transforms, such as the Pairwise Correlating Transform \cite{PCT_Wang}, approaches based on channel coding, such as Unequal Error Protection \cite{UEP_Mohr} and approaches based on quantization, notably the Multiple Description Scalar Quantizer (MDSQ) \cite{mdsq}.
\begin{figure}
\centerline{\includegraphics[width=.9\columnwidth]{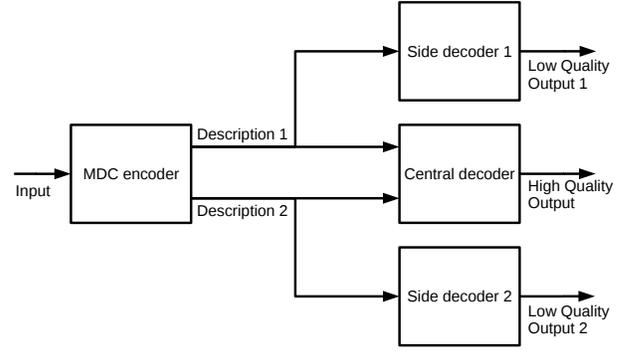}}
\caption{Block diagram of a two-description MDC system.}
\label{fig:mdc}
\end{figure}

\section{Graded quantization formulation}
\label{cs_gq}

\begin{figure}
\centerline{\includegraphics[width=.9\columnwidth]{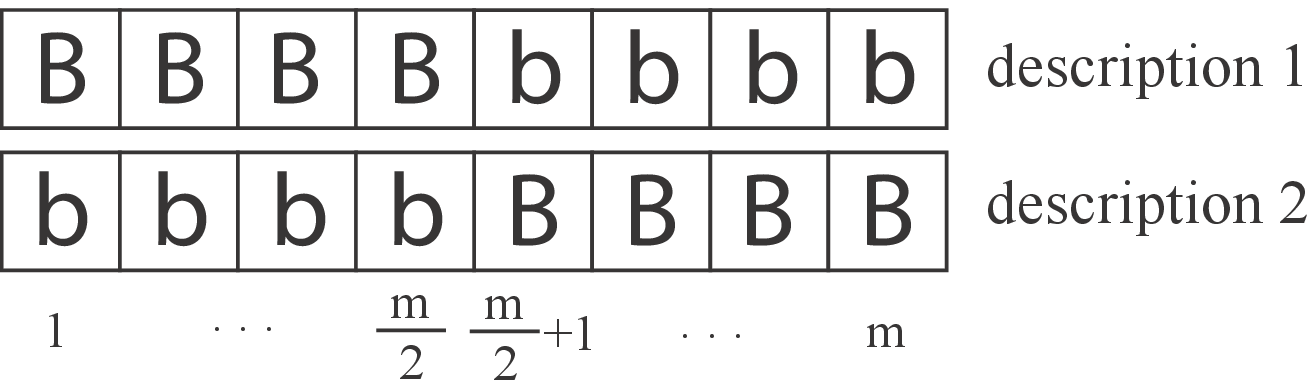}}
\caption{CS-GQ}
\label{GQ_basics}
\end{figure}

In the simple case of two descriptions, CS-GQ partitions the indices of the measurement vector $\mathbf{y}$ into two sets $\Omega_1$ and $\Omega_2$. Two different quantization step sizes (coarse and fine) are chosen and assigned to each set to generate the first description. The dual assignment is performed for the second description. This is graphically shown in Fig.~\ref{GQ_basics}, where we considered $\Omega_1 = \left\lbrace i \in \left[1,\frac{m}{2} \right] \right\rbrace$ and $\Omega_2 = \left\lbrace j \in \left[\frac{m}{2}+1,m \right] \right\rbrace$ but, thanks to the democratic property \cite{Democracy_Laska}, the same performance is expected for any other definition of the sets with the same cardinality, \emph{e.g.}, the sets of even and odd indices. Before going into further details, we remark that this scheme is amenable to generalizations to an arbitrary number of sets and quantization step sizes and hence of descriptions. In this paper we will not consider this more general case, and we will rather focus on the analysis and usage of the two-descriptions system. The proposed system design may provide either balanced or unbalanced descriptions, depending on how many measurements are quantized using the high-rate and the low-rate quantizers. In case of unbalanced descriptions, the cardinality of sets $\Omega_1$ and $\Omega_2$ is not the same. Balanced descriptions offer the same side decoding distortion regardless of which description is received. Otherwise, unbalanced descriptions may have significantly different side distortion depending on the specific description that is received. This is desirable when the descriptions are sent on separate channels having very different packet loss probabilities. In the simplest case, one description is contained in a single packet, \emph{e.g.}, when the number of measurements is small or there is no packet-size limitation. Issues related to packetization strategies and channel losses will be discussed later in Section~\ref{sec:rd_opt}.

From now on, we consider staggered uniform scalar quantizers with $2^B$ and $2^b$ levels, and $B \geq b$. The associated quantization step sizes are, respectively, $\Delta_B$ and $\Delta_b$, which are linked to the number of levels through the dynamic range $r$: $\Delta_B = r2^{-B}$ and $\Delta_b = r2^{-b}$. The $i$-th description, $i=1,2$, has then \mB{i}\ elements quantized with $2^B$ levels and \mb{i}\ elements quantized with $2^b$ levels, such that $\mB{i}+\mb{i}=m$.

Staggered quantizers are quantizers whose reconstruction levels are shifted with respect to each other. In this work, the staggering involves the low-resolution and the high-resolution quantizers and is motivated by an improvement in the resolution obtained by the central decoder, which receives both the high-resolution and low-resolution versions of each measurement. In fact, when both descriptions are received, each measurement falls inside the intersection of the quantization intervals defined by the two staggered quantizers, which is smaller or equal than $\Delta_B$. Hence, the low-resolution quantizer has its bins shifted by $\frac{\Delta_B}{2}$. The following equations describe how the measurement vector $\y$ is quantized to generate the 2 descriptions:

\begin{align*}
\mathbf{y}^{(B)} &= \left\lfloor \frac{\mathbf{y}}{\Delta_B} \right\rfloor \Delta_B + \frac{\Delta_B}{2}\\
\mathbf{y}^{(b)} &= \left\lfloor \frac{\mathbf{y}-\frac{\Delta_B}{2}}{\Delta_b} \right\rfloor \Delta_b + \frac{\Delta_b}{2} + \frac{\Delta_B}{2}~.
\end{align*}

This allows to gain some extra bins of width $\frac{\Delta_B}{2}$ when low-resolution and high-resolution quantizers are combined. This is shown in Fig. \ref{fig:quantizers} for a 4-bit high resolution quantizer and a 2-bit low-resolution quantizer. In general, the high resolution quantizer has $2^B$ bins, the low resolution quantizer has $2^b$ bins, while the combined quantizer has $2^B-2^b+1$ bins of size $\Delta_B$ and $2(2^b-1)$ bins of size $\frac{\Delta_B}{2}$.

\begin{figure}
\centerline{\includegraphics[width=.9\columnwidth]{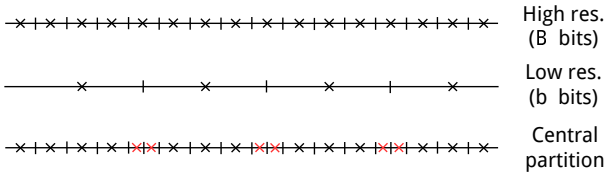}}
\caption{High and low resolution quantizers are staggered so that the central decoder achieves higher precision.}
\label{fig:quantizers}
\end{figure}

In this work we do not consider more complex quantizers, \emph{e.g.}, the Lloyd-Max method or vector quantization, because they are regarded as computationally too complex and with little or no performance gain, as shown in \cite{Milenkovic}.

An improvement over classic recovery from noisy measurements  \eqref{eq:CS_recovery_relaxed} can be obtained if we explicitly take into account the two different quantization noise levels. Hence, the side decoder solves the following optimization problem:
\begin{align}
\hspace*{0.015cm}\thetaside=\arg\min_{\thetav}\left\Vert \thetav\right\Vert _{1}\hspace{0.15cm}\text{s.t.}\hspace{0.05cm}\begin{cases}
\left\Vert \mathbf{y}^{(B)}-\Phi^{(B)}\Psi\thetav\right\Vert _{2}&\leq\varepsilon_{B}\\
\left\Vert \mathbf{y}^{(B)}-\Phi^{(B)}\Psi\thetav\right\Vert _{\infty}&\leq\frac{\Delta_{B}}{2}\\
\left\Vert \mathbf{y}^{(b)}-\Phi^{(b)}\Psi\thetav\right\Vert _{2}&\leq\varepsilon_{b}\\
\left\Vert \mathbf{y}^{(b)}-\Phi^{(b)}\Psi\thetav\right\Vert _{\infty}&\leq\frac{\Delta_{b}}{2}
\end{cases}
\label{sideGQ}
\end{align}
where $\Phi^{(B)}$ and $\Phi^{(b)}$ are the appropriate submatrices of $\Phi$, \emph{i.e.}, $\Phi$ restricted to the rows corresponding to the measurements with fine and coarse quantization levels, respectively. The two $\ell_2$-norm constraints in \eqref{sideGQ} take into account the different quantization levels in the two sets, being $\varepsilon_{B}$ and $\varepsilon_{b}$ the expected $\ell_2$ norm of quantization noise, which can be estimated as
$$
\varepsilon_{B}=\sqrt{\mB{i}\frac{\Delta_B^2}{12}}\qquad \text{and}\qquad \varepsilon_{b}=\sqrt{\mb{i}\frac{\Delta_b^2}{12}}~.
$$

The $\ell_\infty$-norm constraints enforce consistency with the quantization intervals - \emph{i.e.}, a reconstruction whose measurements would be quantized to the same intervals as the observed measurements, allowing better dequantization performance \cite{Democracy_Laska, goyal1998quantized} -, thus exploiting all the information available regarding the original unknown signal. Note that, while the $\ell_2$-norm constraints refer to the \emph{average} norm of the noise, the $\ell_\infty$-norm constraint is applied individually on every linear measurement. In \cite{ValsesiaICASSP}, it is shown that explicitly considering the previously explained structure of the noise in the reconstruction problem provides significant gains in the quality of the reconstruction, with respect to considering an average noise norm.

On the other hand, when both descriptions are received, the central decoder outputs the vector
\begin{equation}
\thetacentral=\arg\min_{\thetav}\left\Vert \thetav\right\Vert _{1}\ \text{s.t.}\ \left\Vert \mathbf{y}_\mathsf{C}-\Phi\Psi\thetav\right\Vert _{2}\leq\varepsilon_\mathsf{C}~,
\label{centralGQ}
\end{equation}
where $\mathbf{y}_\mathsf{C}$ is obtained in the following manner. For each measurement, the quantization bins of the high- and low-resolution versions are compared to determine the correct bin in the central partition to be used for the dequantization (see Fig.~\ref{fig:quantizers}). $\varepsilon_\mathsf{C}^2$ can be estimated as in \eqref{eq:central_quant} to properly take quantizer staggering into account.

We shall also consider a special case of CS-GQ, which we call CS-SPLIT. It is a simple technique that consists in splitting the measurements vector into 2 or more segments, so that instead of transmitting a single packet with all the measurements, two descriptions are created by packetizing half of the measurements in each. Referring to Figure~\ref{GQ_basics}, CS-SPLIT can be regarded as a special case of CS-GQ when $b=0$, providing the best central distortion but worst side distortion as shown in section \ref{GQ_opt}.

\section{ADMM-based side decoder}
\label{admm}
In this section, we propose a method for the resolution of problem \eqref{sideGQ} based on ADMM. In \cite{ValsesiaICASSP} problem \eqref{sideGQ} is solved by CVX \cite{cvx} \cite{cvx2}, a framework to model and solve convex optimization problems, which employs the semidefinite programming (SDP) solver SeDuMi \cite{sedumi}. SeDuMi uses interior point methods to solve general SDP problems, resulting in rather slow performance, especially for large scale problems due to the fact that those are second-order methods\footnote{See \cite{monteiro2003first} for details.}. In this paper, we propose a significantly faster reconstruction algorithm based on the Alternating Directions Method of Multipliers (ADMM) tailored to solve problem \eqref{sideGQ}. ADMM \cite{boydadmm} is a popular first-order method that allows cheaper iterations, and that combines the robustness of the method of multipliers (a method to solve a constrained optimization problem via unconstrained minimization of an augmented Lagrangian function, including the constraints as penalty terms and an additional $\ell_2$ penalty as augmenting term)\cite{methodmultiplier} in terms of convergence, and the decomposition property, similar to dual ascent \cite{boyd2004convex}, which also allows for a distributed implementation. It is particularly convenient to solve the side decoding problem \eqref{sideGQ} because it allows to use alternating updates to the dual variables of an augmented Lagrangian, projecting over the sets defined by the multiple constraints. A proximal gradient method is used to update the primal variable $\thetav$ at each iteration. In the following we report the derivation of the algorithm, while the complete procedure is summarized in Algorithm \ref{decoding_algo}.

\subsection{Notation}
Let us first introduce some notation that is used during the derivation of the algorithm. First, we define the sets $\mathcal{C}_2$ and $\mathcal{C}_\infty$, which are feasible sets defined by the $\ell_2$ and $\ell_\infty$ constraints of \eqref{sideGQ}, as a function of generic quantities $A$, $\mathbf{b}$ and $c$, as
\begin{align*}
\mathcal{C}_2\left( A, \mathbf{b}, c\right) &= \left\lbrace \thetav \in \mathbb{R}^n \ \vert \ \Vert A\thetav-\mathbf{b} \Vert_2 \leq c \right\rbrace\\
\mathcal{C}_\infty\left( A, \mathbf{b}, c\right) &= \left\lbrace \thetav \in \mathbb{R}^n \ \vert \ \Vert A\thetav-\mathbf{b} \Vert_\infty \leq c \right\rbrace~,
\end{align*}
along with the corresponding indicator functions, namely
\begin{align*}
\indfun_{\mathcal{C}_2\left( A, \mathbf{b}, c\right)}\left( \thetav \right) &= \begin{cases}
0 \quad \text{if } \thetav \in \mathcal{C}_2\left( A, \mathbf{b}, c\right) \\
\infty \quad \text{otherwise}
\end{cases}\\
\indfun_{\mathcal{C}_\infty\left( A, \mathbf{b}, c\right)}\left( \thetav \right) &= \begin{cases}
0 \quad \text{if } \thetav \in \mathcal{C}_\infty\left( A, \mathbf{b}, c\right)\\
\infty \quad \text{otherwise}
\end{cases}~.
\end{align*}
The quantities defined above will be used to replace the constrained formulaton of problem \eqref{sideGQ} with an unconstrained one. Moreover, we define the soft-thresholding operator with threshold $\lambda$ as
$$
\mathcal{S}_\lambda \left( \mathbf{v} \right) = \text{sgn}(\mathbf{v})\odot\text{max}\left\lbrace \mathbf{0}, \vert \mathbf{v} \vert - \lambda \right\rbrace~,
$$
where $\vert \cdot \vert$ denotes the vector containing the absolute values of $\cdot$ and $\odot$ denotes the elementwise product, and the operators performing projection over $\ell_2$ ball
$$
\mathcal{P}_{\mathcal{C}_2\left( \mathbf{b}, c\right)} \left( \mathbf{v} \right) = \begin{cases}
\mathbf{b} \!+\! c\frac{\mathbf{v}\!-\!\mathbf{b}}{\Vert \mathbf{v}\!-\!\mathbf{b} \Vert_2} \hspace*{2mm} \text{if } \Vert \mathbf{v}\!-\!\mathbf{b} \Vert_2 \!>\! c\\
\mathbf{v} \quad \text{otherwise}
\end{cases}
$$
and $\ell_\infty$ ball
$$
\text{clip}_{\left[ -c , +c \right]} \left( \mathbf{v} \right) = \text{sgn}(\mathbf{v})\odot\min\left\lbrace \vert \mathbf{v} \vert, c \right\rbrace~,
$$
where $\min\{\cdot\}$ and $\max\{\cdot\}$ have to be intended as  elementwise operators.

\subsection{Unconstrained problem and auxiliary variables}
It is possible to use indicator functions to transform \eqref{sideGQ} into an unconstrained minimization problem. First, we define the following quantities
\begin{align*}
H = \Phi^{(B)}\Psi\quad&\text{and}\quad\mathbf{h} = \mathbf{y}^{(B)}\\
L = \Phi^{(b)}\Psi\quad&\text{and}\quad\mathbf{l} = \mathbf{y}^{(b)}~.
\end{align*}
Then,
\begin{align*}
\thetaside = \arg\underset{\thetav}{\min} \Big\lbrace &\left\Vert \thetav\right\Vert _{1} + \indfun_{\mathcal{C}_2\left( \varepsilon_B, H, \mathbf{h} \right)}\left( \thetav \right) + \indfun_{\mathcal{C}_2\left( \varepsilon_b, L, \mathbf{l} \right)}\left( \thetav \right)\\ &+ \indfun_{\mathcal{C}_\infty\left( \varepsilon_B, H, \mathbf{h} \right)}\left( \thetav \right) + \indfun_{\mathcal{C}_\infty\left( \varepsilon_b, L, \mathbf{l} \right)}\left( \thetav \right) \Big\rbrace \nonumber \\
=\arg\underset{\thetav}{\min} \Big\lbrace &\left\Vert \thetav\right\Vert _{1} + \indfun_{\mathcal{C}_2\left( \varepsilon_B, I, \mathbf{h} \right)}\left( H\thetav \right) + \indfun_{\mathcal{C}_2\left( \varepsilon_b, I, \mathbf{l} \right)}\left( L\thetav \right)\\ &+ \indfun_{\mathcal{C}_\infty\left( \frac{\Delta_B}{2}, I, \mathbf{h} \right)}\left( H\thetav \right) + \indfun_{\mathcal{C}_\infty\left( \frac{\Delta_b}{2}, I, \mathbf{l} \right)}\left( L\thetav \right) \Big\rbrace
\end{align*}
where $I$ is the identity matrix of suitable size. We now introduce auxiliary variables $\mathbf{w}, \mathbf{z}, \mathbf{p}, \mathbf{q}$.
\begin{align*}
\mathbf{w}=H\thetav, \quad \mathbf{z}=L\thetav, \quad \mathbf{p}=H\thetav, \quad \mathbf{q}=L\thetav
\end{align*}

The problem can be now recast as:
\begin{align*}
\hat{\thetav} = \arg\underset{\thetav,\mathbf{w},\mathbf{z},\mathbf{p},\mathbf{q}}{\min} \Big\lbrace &\left\Vert \thetav\right\Vert _{1} + \indfun_{\mathcal{C}_2\left( \varepsilon_B, I, \mathbf{h} \right)}\left( \mathbf{w} \right) + \indfun_{\mathcal{C}_2\left( \varepsilon_b, I, \mathbf{l} \right)}\left( \mathbf{z} \right)\nonumber\\
&+ \indfun_{\mathcal{C}_\infty\left(
 \frac{\Delta_B}{2}, I, \mathbf{h} \right)}\left( \mathbf{p} \right) + \indfun_{\mathcal{C}_\infty\left( \frac{\Delta_b}{2}, I, \mathbf{l} \right)}\left( \mathbf{q} \right) \Big\rbrace \nonumber \\
\text{subject to} \qquad &\mathbf{w} = H\thetav,\quad \mathbf{z} = L\thetav, \quad \mathbf{p} = H\thetav, \quad \mathbf{q} = L\thetav
\end{align*}

We can include these new constraints by building an augmented Lagrangian functional:
\begin{align*}
J = \left\Vert \thetav\right\Vert _{1} &+ \indfun_{\mathcal{C}_2\left( \varepsilon_B, I, \mathbf{h} \right)}\left( \mathbf{w} \right) + \indfun_{\mathcal{C}_2\left( \varepsilon_b, I, \mathbf{l} \right)}\left( \mathbf{z} \right) + \indfun_{\mathcal{C}_\infty\left( \frac{\Delta_B}{2}, I, \mathbf{h} \right)}\left( \mathbf{p} \right) \nonumber\\
& + \indfun_{\mathcal{C}_\infty\left( \frac{\Delta_b}{2}, I, \mathbf{l} \right)}\left( \mathbf{q} \right) + 
\mathbf{u}^T\left( \mathbf{w}-H\thetav \right) +
\mathbf{v}^T\left( \mathbf{z}-L\thetav \right)\nonumber\\
&+ \mathbf{s}^T\left( \mathbf{p}-H\thetav \right) +
\mathbf{t}^T\left( \mathbf{q}-L\thetav \right) + 
\frac{\rho}{2}\Vert \mathbf{w}-H\thetav \Vert_2^2 \nonumber \\
&+ \frac{\rho}{2}\Vert \mathbf{z}-L\thetav \Vert_2^2 + 
\frac{\rho}{2}\Vert \mathbf{p}-H\thetav \Vert_2^2 + 
\frac{\rho}{2}\Vert \mathbf{q}-L\thetav \Vert_2^2~.
\end{align*}
Moreover, we switch to the scaled form by defining the following residuals and scaled dual variables:
\begin{align*}
\mathbf{r}_u&=\mathbf{w}\!-\!H\thetav & \mathbf{r}_v&=\mathbf{z}\!-\!L\thetav & \mathbf{r}_s&=\mathbf{p}\!-\!H\thetav & \mathbf{r}_t&=\mathbf{q}\!-\!L\thetav \\
\mathbf{u}_\rho&=\rho^{-1}\mathbf{u} & \mathbf{v}_\rho&=\rho^{-1}\mathbf{v} & \mathbf{s}_\rho&=\rho^{-1}\mathbf{s}&  \mathbf{t}_\rho&=\rho^{-1}\mathbf{t}
\end{align*}
Hence, the augmented Lagrangian can be rewritten as:
\begin{align}
J = \left\Vert \thetav\right\Vert _{1} &+ \indfun_{\mathcal{C}_2\left( \varepsilon_B, I, \mathbf{h} \right)}\left( \mathbf{w} \right) + \indfun_{\mathcal{C}_2\left( \varepsilon_b, I, \mathbf{l} \right)}\left( \mathbf{z} \right) + \indfun_{\mathcal{C}_\infty\left( \frac{\Delta_B}{2}, I, \mathbf{h} \right)}\left( \mathbf{p} \right)\nonumber\\
&+ \indfun_{\mathcal{C}_\infty\left( \frac{\Delta_b}{2}, I, \mathbf{l} \right)}\left( \mathbf{q} \right) +
\frac{\rho}{2}\Vert \mathbf{r}_u+\mathbf{u}_\rho \Vert_2^2 - \frac{\rho}{2}\Vert \mathbf{u}_\rho \Vert_2^2 \nonumber\\
&+ \frac{\rho}{2}\Vert \mathbf{r}_v+\mathbf{v}_\rho \Vert_2^2 - \frac{\rho}{2}\Vert \mathbf{v}_\rho \Vert_2^2
+ \frac{\rho}{2}\Vert \mathbf{r}_s+\mathbf{s}_\rho \Vert_2^2 \nonumber\\
&- \frac{\rho}{2}\Vert \mathbf{s}_\rho \Vert_2^2 + 
\frac{\rho}{2}\Vert \mathbf{r}_t+\mathbf{t}_\rho \Vert_2^2 - \frac{\rho}{2}\Vert \mathbf{t}_\rho \Vert_2^2~. \label{augmentedLagrange} 
\end{align}

\subsection{Alternating directions minimization}

\begin{algorithm*}
\begin{algorithmic}
\Require $\alpha, \lambda, \rho, H, L, \mathbf{h}, \mathbf{l}$
\State $\thetav^{(0,0)}\gets\mathbf{0}, \mathbf{w}^{(0)}=\mathbf{p}^{(0)} \gets \mathbf{0}, \mathbf{z}^{(0)}=\mathbf{q}^{(0)}\gets\mathbf{0}, \mathbf{u}_\rho^{(0)}=\mathbf{s}_\rho^{(0)} \gets \mathbf{0} , \mathbf{v}_\rho^{(0)}=\mathbf{t}_\rho^{(0)}\gets\mathbf{0}$
\While{Outer stopping criterion is not met}
\While{Inner stopping criterion is not met}
\begin{align*}
\hspace*{-2.8cm}
\thetav^{(i,j+1)} = \mathcal{S}_\lambda \Bigl( \thetav^{(i,j)} - \alpha \rho &\Bigl[  H^T\left( H\thetav^{(i,j)}-\mathbf{w^{(i)}}-\mathbf{u}_\rho^{(i)} \right)  \nonumber +  L^T\left( L\thetav^{(i,j)}-\mathbf{z}^{(i)}-\mathbf{v}_\rho^{(i)} \right) \nonumber\\&+ H^T\left( H\thetav^{(i,j)}-\mathbf{p}^{(i)}-\mathbf{s}_\rho^{(i)} \right) \nonumber + L^T\left( L\thetav^{(i,j)}-\mathbf{q}^{(i)}-\mathbf{t}_\rho^{(i)} \right) \Bigr] \Bigr)
\end{align*}
\State $j \gets j+1$
\EndWhile
\State After stopping: $\thetav^{(i+1,0)}=\thetav^{(i,j)}$
\begin{flalign}
\hspace*{0.5cm}
&\mathbf{w}^{(i+1)} = \begin{cases}
\mathbf{h} + \varepsilon_B\frac{H\thetav^{(i+1,0)} - \mathbf{u}_\rho^{(i)}-\mathbf{h}}{\Vert H\thetav^{(i+1,0)} - \mathbf{u}_\rho^{(i)}-\mathbf{h} \Vert_2} \quad \text{if } \Vert H\thetav^{(i+1,0)} - \mathbf{u}_\rho^{(i)}-\mathbf{h} \Vert_2 > \varepsilon_B\\
H\thetav^{(i+1,0)} - \mathbf{u}_\rho^{(i)} \quad \text{otherwise}
\end{cases} & \label{wupdate}&\\
&\mathbf{z}^{(i+1)} = \begin{cases}
\mathbf{l} + \varepsilon_b\frac{L\thetav^{(i+1,0)} - \mathbf{v}_\rho^{(i)}-\mathbf{l}}{\Vert L\thetav^{(i+1,0)} - \mathbf{v}_\rho^{(i)}-\mathbf{l} \Vert_2} \quad \text{if } \Vert L\thetav^{(i+1,0)} - \mathbf{v}_\rho^{(i)}-\mathbf{l} \Vert_2 > \varepsilon_b\\
L\thetav^{(i+1,0)} - \mathbf{v}_\rho^{(i)} \quad \text{otherwise}
\end{cases} & \label{zupdate}&\\
&\mathbf{p}^{(i+1)} = \text{clip}_{\left[ -\frac{\Delta_B}{2} , \frac{\Delta_B}{2} \right]} \left( H\thetav^{(i+1,0)} - \mathbf{s}_\rho^{(i)} - \mathbf{h} \right) + \mathbf{h} & \label{pupdate}&\\
&\mathbf{q}^{(i+1)} = \text{clip}_{\left[ -\frac{\Delta_b}{2} , \frac{\Delta_b}{2} \right]} \left( L\thetav^{(i+1,0)} - \mathbf{t}_\rho^{(i)} - \mathbf{l} \right) + \mathbf{l} & \label{qupdate}&\\
&\mathbf{u}_\rho^{(i+1)} = \mathbf{u}_\rho^{(i)} + \mathbf{w}^{(i+1)} - H\thetav^{(i+1,0)} \nonumber&\\
&\mathbf{v}_\rho^{(i+1)} = \mathbf{v}_\rho^{(i)} + \mathbf{z}^{(i+1)} - L\thetav^{(i+1,0)} \nonumber&\\
&\mathbf{s}_\rho^{(i+1)} = \mathbf{s}_\rho^{(i)} + \mathbf{p}^{(i+1)} - H\thetav^{(i+1,0)} \nonumber&\\
&\mathbf{t}_\rho^{(i+1)} = \mathbf{t}_\rho^{(i)} + \mathbf{q}^{(i+1)} - L\thetav^{(i+1,0)} \nonumber&
\end{flalign}
\State $i \gets i+1$
\EndWhile
\Ensure $\thetaside \gets \thetav^{(i,0)}$ 

\end{algorithmic}
\caption{CS-GQ side decoder.}
\label{decoding_algo}
\end{algorithm*}

It is now possible to minimize \eqref{augmentedLagrange} in an iterative fashion using alternating directions, \emph{i.e.}, minimizing over one variable at a time. This constitutes the outer loop of Algorithm~\ref{decoding_algo}.
\begin{align}
\thetav^{(i+1)} &= \arg\underset{\thetav}{\min} \Big\lbrace \left\Vert \thetav\right\Vert _{1} \nonumber\\&+ 
\frac{\rho}{2}\Vert \mathbf{w}^{(i)}-H\thetav+\mathbf{u}_\rho^{(i)} \Vert_2^2 + 
\frac{\rho}{2}\Vert \mathbf{z}^{(i)}-L\thetav+\mathbf{v}_\rho^{(i)} \Vert_2^2  \nonumber\\ &+
\frac{\rho}{2}\Vert \mathbf{p}^{(i)}-H\thetav+\mathbf{s}_\rho^{(i)} \Vert_2^2  + 
\frac{\rho}{2}\Vert \mathbf{q}^{(i)}-L\thetav+\mathbf{t}_\rho^{(i)} \Vert_2^2  \Big\rbrace \label{mintheta}\\
\mathbf{w}^{(i+1)} &= \arg\underset{\mathbf{w}}{\min} \Big\lbrace \indfun_{\mathcal{C}_2\left( \varepsilon_B, I, \mathbf{h} \right)}\left( \mathbf{w} \right) \nonumber\\&+ \frac{\rho}{2}\Vert \mathbf{w}-H\thetav^{(i+1)}+\mathbf{u}_\rho^{(i)} \Vert_2^2 \Big\rbrace \label{minw}\\
\mathbf{z}^{(i+1)} &= \arg\underset{\mathbf{z}}{\min} \Big\lbrace \indfun_{\mathcal{C}_2\left( \varepsilon_b, I, \mathbf{l} \right)}\left( \mathbf{z} \right) \nonumber\\&+ \frac{\rho}{2}\Vert \mathbf{z}-L\thetav^{(i+1)}+\mathbf{v}_\rho^{(i)} \Vert_2^2 \Big\rbrace \label{minz}\\
\mathbf{p}^{(i+1)} &= \arg\underset{\mathbf{p}}{\min} \Big\lbrace \indfun_{\mathcal{C}_\infty\left( \frac{\Delta_B}{2}, I, \mathbf{h} \right)}\left( \mathbf{p} \right) \nonumber\\&+ \frac{\rho}{2}\Vert \mathbf{p}-H\thetav^{(i+1)}+\mathbf{s}_\rho^{(i)} \Vert_2^2 \Big\rbrace \label{minp}
\end{align}
\begin{align}
\mathbf{q}^{(i+1)} &= \arg\underset{\mathbf{q}}{\min} \Big\lbrace \indfun_{\mathcal{C}_\infty\left( \frac{\Delta_b}{2}, I, \mathbf{l} \right)}\left( \mathbf{q} \right) \nonumber\\&+ \frac{\rho}{2}\Vert \mathbf{q}-L\thetav^{(i+1)}+\mathbf{t}_\rho^{(i)} \Vert_2^2 \Big\rbrace \label{minq}\\
\mathbf{u}_\rho^{(i+1)} &= \mathbf{u}_\rho^{(i)} + \mathbf{w}^{(i+1)} - H\thetav^{(i+1)}\nonumber\\
\mathbf{v}_\rho^{(i+1)} &= \mathbf{v}_\rho^{(i)} + \mathbf{z}^{(i+1)} - L\thetav^{(i+1)}\nonumber\\
\mathbf{s}_\rho^{(i+1)} &= \mathbf{s}_\rho^{(i)} + \mathbf{p}^{(i+1)} - H\thetav^{(i+1)}\nonumber\\
\mathbf{t}_\rho^{(i+1)} &= \mathbf{t}_\rho^{(i)} + \mathbf{q}^{(i+1)} - L\thetav^{(i+1)}\nonumber
\end{align}

It can be seen that \eqref{mintheta} involves minimizing a functional composed of two main parts: a smooth  (differentiable) part $f(\thetav) = \frac{\rho}{2}\Vert \mathbf{w}^{(i)}-H\thetav+\mathbf{u}_\rho^{(i)} \Vert_2^2 + 
\frac{\rho}{2}\Vert \mathbf{z}^{(i)}-L\thetav+\mathbf{v}_\rho^{(i)} \Vert_2^2  +
\frac{\rho}{2}\Vert \mathbf{p}^{(i)}-H\thetav+\mathbf{s}_\rho^{(i)} \Vert_2^2  + 
\frac{\rho}{2}\Vert \mathbf{q}^{(i)}-L\thetav+\mathbf{t}_\rho^{(i)} \Vert_2^2$ whose gradient can be computed analytically and a non-smooth (non-differentiable) part $g(\thetav)=\left\Vert \thetav\right\Vert _{1}$. Functionals of this kind can be minimized using the proximal gradient method \cite{proximal}, \emph{i.e.}, an iterative method that takes a step towards the negative gradient of the smooth component and then computes the proximal mapping over the non-smooth part. It is easy to show that the gradient of the smooth part $f(\thetav)$ is:
\begin{align*}
&\nabla_{\thetav} f(\thetav) = \\
& = \rho H^T \left( H\thetav^{(i)} \!-\! \mathbf{w}^{(i)} \!-\! \mathbf{u}_\rho^{(i)} \right) +
\rho L^T \left( L\thetav^{(i)} \!-\! \mathbf{z}^{(i)} \!-\! \mathbf{v}_\rho^{(i)} \right) \\ &+
\rho H^T \left( H\thetav^{(i)} \!-\! \mathbf{p}^{(i)} \!-\! \mathbf{s}_\rho^{(i)} \right) +
\rho L^T \left( L\thetav^{(i)} \!-\! \mathbf{q}^{(i)} \!-\! \mathbf{t}_\rho^{(i)} \right)~,
\end{align*}
while the proximity operator of non-smooth part $g(\thetav)$ is the soft-thresholding function.
Hence, one can iteratively take a step in the gradient direction and apply the soft thresholding function in order to eventually find $\thetav^{(i+1)}$, \emph{i.e.}, the minimizer of \eqref{mintheta}. This constitutes the inner loop of Algorithm~\ref{decoding_algo}.

The subsequent minimization problems in \eqref{minw}, \eqref{minz}, \eqref{minp} and \eqref{minq} all involve projections over the sets defined using the indicator functions. It can be seen that \eqref{minw} and \eqref{minz} correspond to projections over the $\ell_2$ balls of radii $\varepsilon_B$ and $\varepsilon_b$ centered at $\mathbf{h}$ and $\mathbf{l}$ respectively. Similarly, \eqref{minp} and \eqref{minq} correspond to projections over the $\ell_\infty$ balls of diameters $\Delta_B$ and $\Delta_b$ centered at $\mathbf{h}$ and $\mathbf{l}$ respectively. The results of the projections are \eqref{wupdate}, \eqref{zupdate}, \eqref{pupdate} and \eqref{qupdate}. Notice that the projection operations and the update of the dual variables can be performed in parallel.

As stopping criteria of the inner and outer loops, we check the distance between two successive iterations. Hence, two constants $\varepsilon^{\mathsf{stop}}_i$ and $\varepsilon^{\mathsf{stop}}_j$ are properly chosen such that the stopping conditions are
\begin{align}
\frac{\Vert \thetav^{(i,j+1)} - \thetav^{(i,j)} \Vert_2}{\Vert \thetav^{(i,j)} \Vert_2} &< \varepsilon^{\mathsf{stop}}_j \tag{Inner stopping criterion} \\
\frac{\Vert \thetav^{(i+1,0)} - \thetav^{(i,0)} \Vert_2}{\Vert \thetav^{(i,0)} \Vert_2} &< \varepsilon^{\mathsf{stop}}_i \tag{Outer stopping criterion}.
\end{align}

\section{Theoretical guarantees}
\label{sec:theory}
In this section we present some theoretical results concerning CS-GQ. First, Theorem~\ref{stability} provides a guarantee of stable recovery of signals from single descriptions, using the side decoder presented in \eqref{sideGQ}. As it is common in the literature about CS (see for example \cite{davenport2012introduction}), this kind of bounds is used as a guarantee that the error does not explode, rather than giving an exact characterization of the error itself.

\begin{theorem}
\label{stability}
\textbf{\emph{(Stable side recovery)}}
Suppose that $A=\Phi\Psi$ satisfies the RIP of order $2k$ with $\delta_{2k}<\sqrt{2}-1$ and let $\mathbf{y} = A\thetav+\mathbf{e}$ where, without loss of generality, $\mathbf{e}=\left[ \mathbf{e}_B^T \  \mathbf{e}_b^T \right]^T$ with $\Vert \mathbf{e}_B \Vert_2 \leq \varepsilon_B$, $\Vert \mathbf{e}_b \Vert_2 \leq \varepsilon_b$, $\Vert \mathbf{e}_B \Vert_\infty \leq \frac{\Delta_B}{2}$, $\Vert \mathbf{e}_b \Vert_\infty \leq \frac{\Delta_b}{2}$. Then the solution to \eqref{sideGQ} obeys
\begin{multline*}
\big\Vert \thetaside - \thetav \big\Vert_2 \leq C_0 \frac{\sigma_k\left(\thetav\right)_1}{\sqrt{k}} \\+ C_2 \left( \min\left\lbrace \varepsilon_B , \frac{\Delta_B}{2}\sqrt{m_B} \right\rbrace + \min\left\lbrace \varepsilon_b , \frac{\Delta_b}{2}\sqrt{m_b} \right\rbrace \right)
\end{multline*}
where $\sigma_k\left(\thetav\right)_1$ is the $l_1$-norm of the error incurred by approximating \thetav\ with its $k$ largest-magnitude components, $m_B$ and $m_b$ are the number of rows of matrices $\Phi^{(B)}\Psi$ and $\Phi^{(b)}\Psi$ respectively and the constants are 
\begin{align*}
C_0 = 2\frac{1-\left( 1-\sqrt{2}\delta_{2k} \right)}{1-\left( 1+\sqrt{2}\delta_{2k} \right)} \qquad C_2 = 4\frac{\sqrt{1+\delta_{2k}}}{1-\left( 1+\sqrt{2}\delta_{2k} \right)}.
\end{align*}
\end{theorem}

\begin{proof}
We are interested in deriving a bound for $\Vert \mathbf{d} \Vert_2 = \big\Vert \thetaside - \thetav \big\Vert_2$. We know that $\thetav \in \left\lbrace \mathcal{C}_2\left( \varepsilon_B, H, \mathbf{h} \right) \cap \mathcal{C}_2\left( \varepsilon_b, L, \mathbf{l} \right) \cap \mathcal{C}_\infty\left( \frac{\Delta_B}{2}, H, \mathbf{h} \right) \cap \mathcal{C}_\infty\left( \frac{\Delta_b}{2}, L, \mathbf{l} \right) \right\rbrace$ (see Sec. \ref{admm} for notation). Moreover, the solution to \eqref{sideGQ} is either $\thetav$ or one with lower $l_1$ norm, so that we can say that $\big\Vert \thetaside \big\Vert_1 \leq \Vert \thetav \Vert_1$. Using \cite[Lemma 1.6]{davenport2012introduction} we know that
\begin{align*}
\Vert \mathbf{d} \Vert_2 \leq C_0 \frac{\sigma_k\left(\thetav\right)_1}{\sqrt{k}} + C_1 \frac{\vert \left\langle A\mathbf{d}_{\Lambda} , A\mathbf{d} \right\rangle \vert}{ \Vert \mathbf{d}_{\Lambda} \Vert_2 }
\end{align*}
where $C_1 = \frac{2}{1-(1+\sqrt{2})\delta_{2k}}$, $\Lambda = \Lambda_0 \cup \Lambda_1$, being $\Lambda_0$ the set of the $k$ components with largest magnitude of $\thetav$ and $\Lambda_1$ the set of the $k$ components with largest magnitude of $\mathbf{d}_{\Lambda_0^c}$ (subscript denotes restriction to the components indexed by the subscript set).
Notice that
\begin{align*}
\Vert A\mathbf{d} \Vert_2 &= \big\Vert A\thetaside - A\thetav \big\Vert_2 = \big\Vert A\thetaside - \mathbf{y} + \mathbf{y} - A\thetav \big\Vert_2\\ & = \Bigg\Vert \left[ \begin{array}{c}
\Phi^{(B)}\Psi\thetaside\\
\Phi^{(b)}\Psi\thetaside
\end{array} \right] - \left[ \begin{array}{c}
\mathbf{y}^{(B)}\\
\mathbf{y}^{(b)}
\end{array} \right] \\
&+ \left[ \begin{array}{c}
\mathbf{y}^{(B)}\\
\mathbf{y}^{(b)}
\end{array} \right] - \left[ \begin{array}{c}
\Phi^{(B)}\Psi\thetav\\
\Phi^{(b)}\Psi\thetav
\end{array} \right]  \Bigg\Vert_2 \\
&\leq \left\Vert \Phi^{(B)}\Psi \thetaside - \mathbf{y}^{(B)} \right\Vert_2 + \left\Vert \Phi^{(b)}\Psi\thetaside - \mathbf{y}^{(b)} \right\Vert_2 \\
&+\left\Vert \Phi^{(B)}\Psi\thetav - \mathbf{y}^{(B)} \right\Vert_2 + \left\Vert \Phi^{(b)}\Psi\thetav - \mathbf{y}^{(b)} \right\Vert_2 \\
&\leq 2\min\left\lbrace \varepsilon_B , \frac{\Delta_B}{2}\sqrt{m_B} \right\rbrace + 2\min\left\lbrace \varepsilon_b , \frac{\Delta_b}{2}\sqrt{m_b} \right\rbrace
\end{align*}
Using the Cauchy-Schwarz inequality we write
\begin{multline*}
\vert \left\langle A\mathbf{d}_{\Lambda} , A\mathbf{d} \right\rangle \vert \leq \Vert A\mathbf{d}_{\Lambda} \Vert_2 \Vert A\mathbf{d} \Vert_2 \\
\leq 2\Vert \mathbf{d}_{\Lambda} \Vert_2 \sqrt{1\!+\!\delta_{2k}} \left( \min\left\lbrace \varepsilon_B , \frac{\Delta_B}{2}\sqrt{m_B} \right\rbrace \!+\! \min\left\lbrace \varepsilon_b , \frac{\Delta_b}{2}\sqrt{m_b} \right\rbrace \right)
\end{multline*}
Thus,
\begin{multline*}
\Vert \mathbf{d} \Vert_2 \leq C_0 \frac{\sigma_k\left(\thetav\right)_1}{\sqrt{k}}\\
+ C_2 \left( \min\left\lbrace \varepsilon_B , \frac{\Delta_B}{2}\sqrt{m_B} \right\rbrace + \min\left\lbrace \varepsilon_b , \frac{\Delta_b}{2}\sqrt{m_b} \right\rbrace \right)
\end{multline*}
\end{proof}

Next, we characterize the performance of an oracle side decoder for sparse signals so that the result can be used to provide an oracle optimality condition for the optimization of the CS-GQ parameters, as shown in Sec. \ref{sec:rd_opt}. The oracle decoder is an ideal decoder that knows perfectly the support $\mathcal{S}$ of the sparse signal. For such systems, recovery amounts to computing $\hat{\thetav} = A_{\mathcal{S}}^{\dagger}\mathbf{y}$, where $A_{\mathcal{S}}^{\dagger}$ is the Moore-Penrose pseudoinverse of matrix $A$ restricted to the columns indexed by $\mathcal{S}$. Often, the oracle receiver is used to evaluate the exact performance of CS reconstruction algorithms, since CS decoding involves a nonlinear reconstruction step, whose distortion performance is hard to characterize exactly.

\begin{theorem}
\textbf{\emph{(Oracle recovery)}}
\label{oracle_thm}
Suppose that $\thetav$ is $k$-sparse and that $\Phi$ has i.i.d. Gaussian zero-mean entries having variance $\frac{1}{m}$. Let $\mathbf{y} = \Phi\Psi\thetav$ be the vector of measurements to be quantized using staggered low- and high-resolution quantizers. Suppose the oracle decoder is used for CS recovery. Then, the expected side distortion in the high-rate regime is:
\begin{align*}
\mathbb{E}\left[ \left\Vert \thetaside - \thetav \right\Vert_2^2 \right] = \frac{kr^2}{\left(m-k-1\right)}\frac{m}{24}\left( 2^{-2B} + 2^{-2b} \right)
\end{align*}
and the expected central distortion is:
\begin{align*}
\mathbb{E}\left[ \left\Vert \thetacentral - \thetav \right\Vert_2^2 \right] = \frac{kr^2}{\left(m-k-1\right)}\frac{m}{24} 2^{-2B} \frac{2^{B+1}-2^b+1}{2^B+2^b-1}
\end{align*}
\end{theorem}
\begin{proof}
According to \cite[Theorem 5]{ColucciaRD}, the following relation holds:
\begin{align}
\label{oracle_coluccia}
\mathbb{E}\left[ \Vert \hat{\thetav} - \thetav \Vert_2^2 \right] = \frac{k}{m-k-1}\mathbb{E}\left[ \mathbf{e}^T\mathbf{e} \right]
\end{align}
where $\mathbf{e} = \mathcal{Q}(\mathbf{y}) - \mathbf{y}$ is the error introduced by quantization $\mathcal{Q}(\cdot)$ 
In case of side decoding, the receiver has only a single description $i$ with \mB{i}\ high resolution measurements and \mb{i}\ low resolution measurements. The average error norm is obtained as the mean of the norm of the expected error when description 1 is received ($\mathbf{e}_1$) and when description 2 is received  ($\mathbf{e}_2$). Notice that $\mB{1} = \mb{2}$ and $\mb{1} = \mB{2}$ by design. Hence, in the high-rate regime the following relation holds:
\begin{align*}
\mathbb{E}\left[ \left\Vert \thetaside - \thetav \right\Vert_2^2 \right] &= \frac{k}{m-k-1}\cdot\frac{1}{2}\left(\mathbb{E}\left[\mathbf{e}_1^{T}\mathbf{e}_1 \right]+\mathbb{E}\left[\mathbf{e}_2^{T}\mathbf{e}_2 \right]\right)\\
&= \frac{kr^2}{m-k-1}\frac{1}{24} \Big( \mB{1} 2^{-2B} + \mb{1} 2^{-2b} \\&+ \mB{2} 2^{-2B} + \mb{2} 2^{-2b} \Big)\\
&= \frac{kr^2}{\left(m-k-1\right)}\frac{m}{24}\left( 2^{-2B} + 2^{-2b} \right)
\end{align*}
The central decoder exploits the staggering of the low and high resolution quantizers to obtain a non-uniform central quantizer, having $2^B-2^b+1$ bins of size $\Delta_B$ and $2(2^b-1)$ bins of size $\frac{\Delta_B}{2}$ as explained in Sec. \ref{cs_gq}. Thus the expected error norm is:
\begin{align}
&\mathbb{E}\left[ \mathbf{e}^T\mathbf{e} \right] =\nonumber\\ &= \frac{m}{2^B+2^b-1} \left[ 2\left( 2^b-1 \right) \frac{\Delta_B^2}{48} + \left( 2^B-2^b+1 \right) \frac{\Delta_B^2}{12} \right]\nonumber\\& = \frac{mr^2}{24}2^{-2B} \frac{2^{B+1}-2^b+1}{2^B+2^b-1}
\label{eq:central_quant}
\end{align}
Substituting it into \eqref{oracle_coluccia} we obtain the desired expression.
\end{proof}

\section{Rate-distortion optimization and packetization}
\label{sec:rd_opt}

\subsection{Optimizing CS-GQ}
\label{GQ_opt}

In this section, we focus our analysis on the scenario when one employs identical channel models for both descriptions, \emph{e.g.}, the loss probability $p$ is the same for both descriptions. Under this scenario, balanced descriptions are optimal both instance-wise and on average, \emph{i.e.} the distortion incurred by either side decoder is the same. On the other hand, unbalanced descriptions provide the same average performance, but the distortion incurred by a specific instance depends on which description is received and can be either lower or higher than the distortion in the balanced case. Hereafter, we thus only consider the balanced case.
Different values of the parameters $B$ and $b$ provide different tradeoffs between the distortion at the central decoder and the distortion at the side decoder, for a fixed total rate for both descriptions: $R=B+b$. Moreover, the packet loss rate on the communication links will also affect the expected distortion.
On the other hand, the limit case $b=0$, which we called CS-SPLIT, simply splits the measurements into two sets without inserting any redundancy. This is equivalent to the segmentation performed by any network protocol when the packet size exceeds the maximum size. It is clear that this is the optimal strategy when there are no packet losses, because it is equivalent to generating a single description and no bits are wasted in redundancy. However, if one description fails to reach the receiver, only half of the measurements will be available for decoding, hence significantly degrading the quality or leading to recovery failure when their number is too low. Fig. \ref{tradeoff} shows the various operating points enabled by CS-GQ on the side distortion vs. central distortion plane.

\begin{figure*}
\centerline{
\subfloat[$m=80$]{
\includegraphics[width=0.43\textwidth]{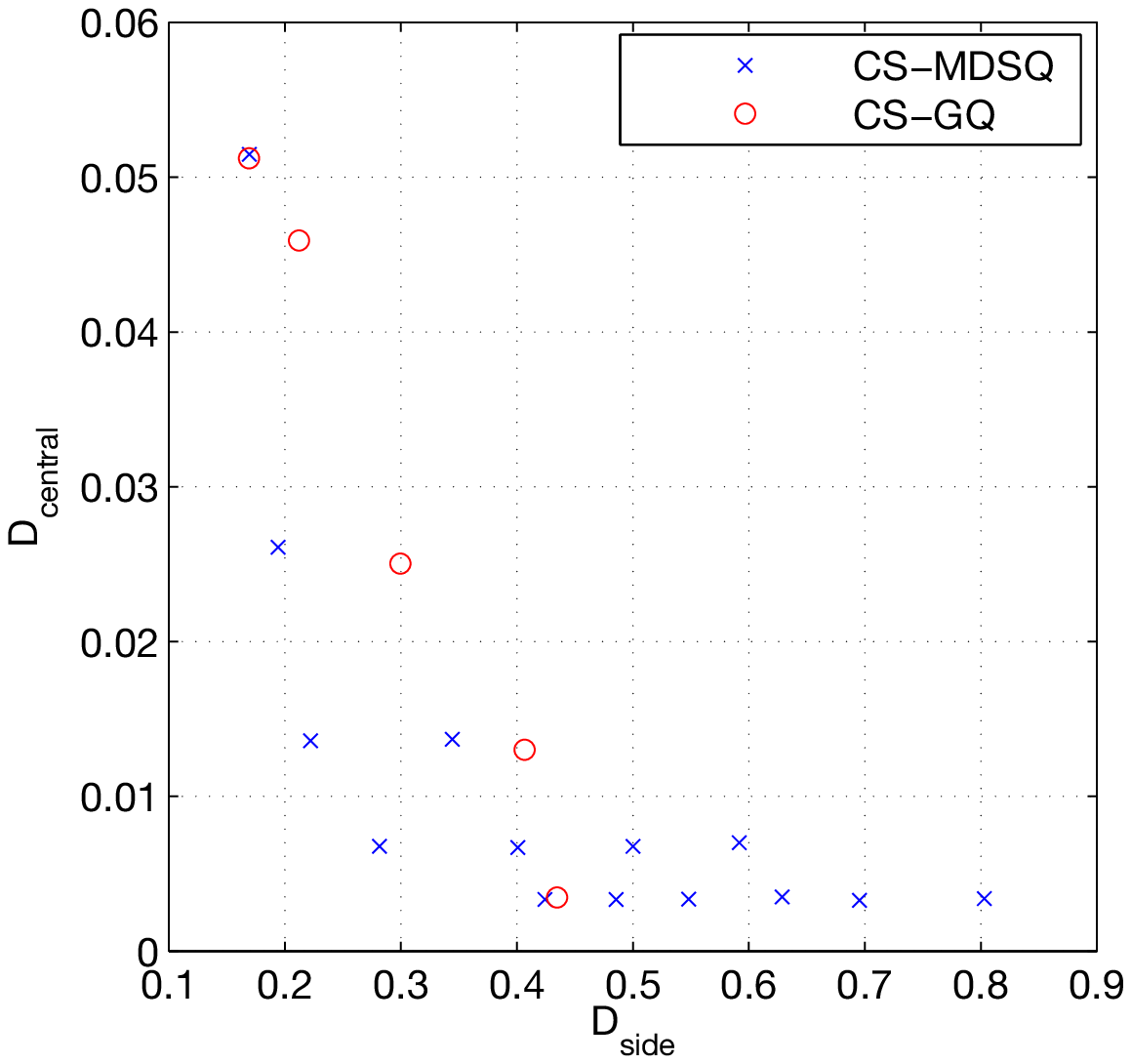}}\qquad
\subfloat[$m=120$]{
\includegraphics[width=0.43\textwidth]{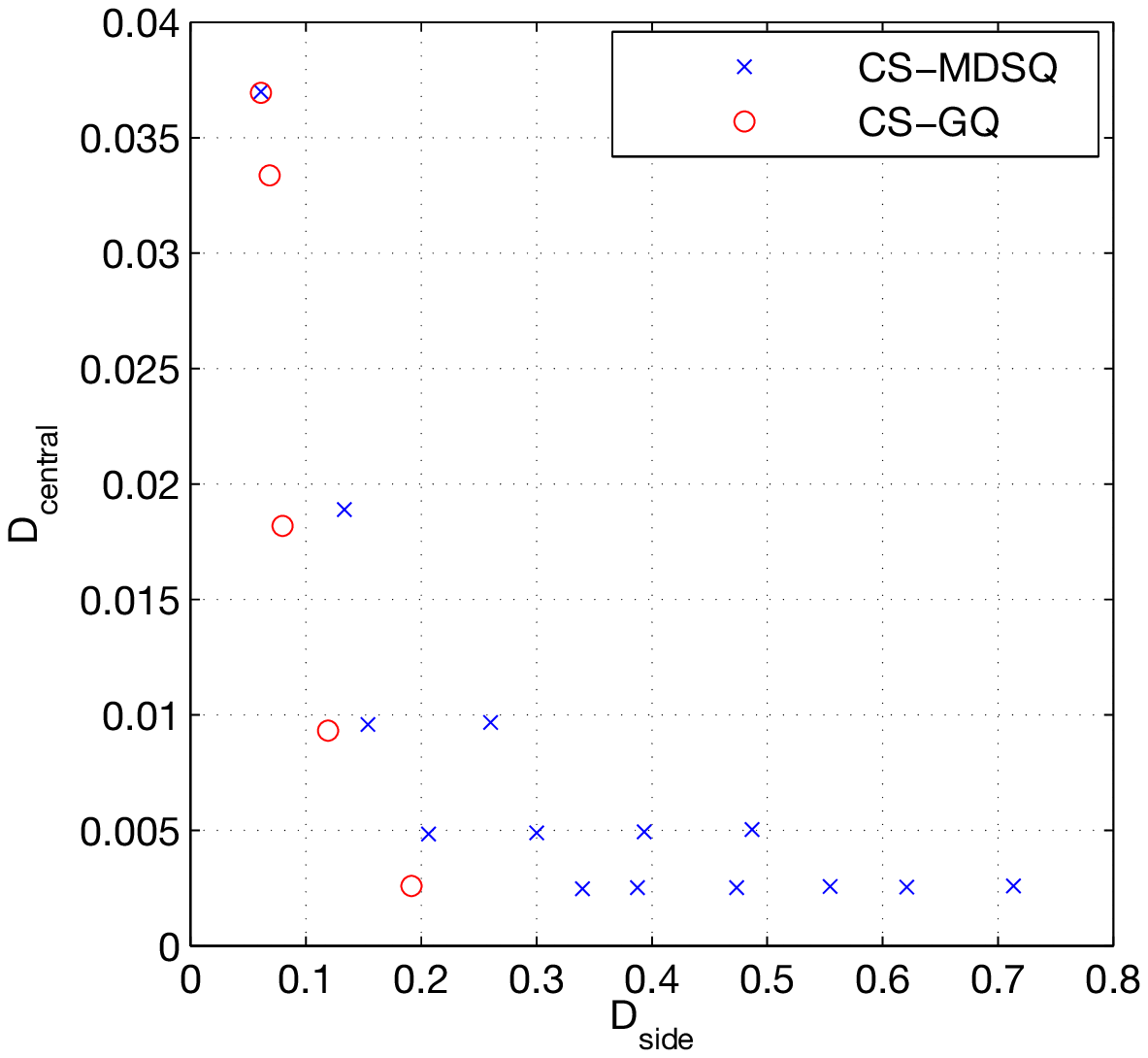}}}
\caption{Operating points for CS-GQ (different values of $B$ and $b$) and for CS-MDSQ (different values of $M$, refer to \cite{mdsq2011}). $n=256$, $k=10$.}
\label{tradeoff}
\end{figure*}

It is of interest to study the optimal value of the parameters $B$ and $b$ when the description loss probability is known. In this scenario, one wishes to find the quantization step sizes providing the lowest expected distortion for a given bit budget. Hence we define the average distortion as
\begin{align*}
D = p^2 + 2p\left(1-p\right) \cdot D_s + \left(1-p\right)^2 \cdot D_c
\label{Davg}
\end{align*}
where $p$ is the probability of losing a description, and $D_s$ and $D_c$ are the distortions incurred by the side and central decoder, respectively defined as
$$
D_s = \frac{\Vert \thetav-\thetaside \Vert_2}{\Vert \thetav \Vert_2}\ \text{and}\ D_c = \frac{\Vert \thetav-\thetacentral \Vert_2}{\Vert \thetav \Vert_2}~.
$$

The main problem in the optimization of the parameters is the lack of closed-form expressions for the distortion of the central and side decoders. There are two possible ways of solving this problem. The first way is to resort to an operational curve, in which the operating points on the ($D_s$,$D_c$) plane are known in advance (\emph{e.g.}, as a result of experiments).

In this case, for a fixed total rate $R$, every choice of the low-resolution rate $b$ generates a $(D_s(b),D_c(b))$ point, for a total of $\lfloor \frac{R}{2} \rfloor+1$ points. The optimal choice of $b$ is therefore:
\begin{align}
\hat{b} = \underset{b\in\left[0,\frac{R}{2}\right]}{\arg\min}\ \left[ p^2+2p\cdot(1-p)D_s(b) + (1-p)^2 \cdot D_c(b) \right] \tag{Operational Optimality}
\end{align}

The other possible method is to employ some bounds on the reconstruction performance. We can use the performance of the oracle decoder, as described in Theorem \ref{oracle_thm}, to derive a simple expression that does not require any operational information and can be used to optimize the parameters \emph{a priori}. Hence, the minimization of the average distortion becomes:
\begin{align*}
\hat{b} = \underset{b\in\left[0,\frac{R}{2}\right]}{\arg\min} &\left[ p^2 + 2p\left(1-p\right) \cdot D_s + \left(1-p\right)^2 \cdot D_c \right] \\
= \underset{b\in\left[0,\frac{R}{2}\right]}{\arg\min} &\Bigg[ 2p\left(1-p\right) \frac{kr^2}{\left(m-k-1\right)}\frac{m}{24}\left( 2^{-2B} + 2^{-2b} \right) \\&+ \left( 1-p \right)^2 \frac{kr^2}{\left(m-k-1\right)}\frac{m}{24} 2^{-2B} \frac{2^{B+1}-2^b+1}{2^B+2^b-1}  \Bigg] \\
= \underset{b\in\left[0,\frac{R}{2}\right]}{\arg\min} &\Bigg[ 2p \left( 2^{-2B} + 2^{-2b} \right) \\&+ (1-p)2^{-2B} \frac{2^{B+1}-2^b+1}{2^B+2^b-1} \Bigg] \\
= \underset{b\in\left[0,\frac{R}{2}\right]}{\arg\min} &\Bigg[ 2p \left( 2^{-2(R-b)} + 2^{-2b} \right) \\&+ (1-p)2^{-2(R-b)} \frac{2^{R-b+1}-2^b+1}{2^{R-b}+2^b-1} \Bigg]
\tag{Oracle Optimality}
\end{align*}

Exhaustive search over the feasible integer values of $b$ can be employed to determine the optimal one. We remark that this search just involves evaluating the distortion expression for $\lfloor \frac{R}{2} \rfloor+1$ values of $b$ (\emph{e.g.}, 5 values when $R=8$), thus having very low complexity.
Unfortunately, the optimality would only hold if we had an oracle decoder, so this choice might be suboptimal. Section \ref{optimal_simulations} shows a comparison between the average distortion obtained with operational optimization and with oracle optimization in order to show the effectiveness of the latter.

\subsection{Packetization for high number of measurements}
\label{gq_large_scale}

\begin{figure}
\centering
\subfloat[Wrong way]{\includegraphics[width=\columnwidth]{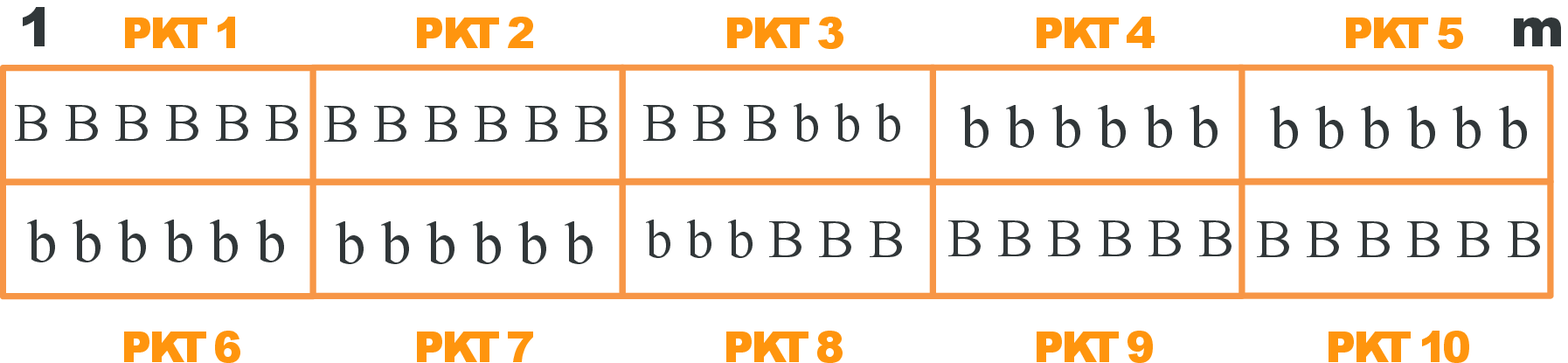}\label{GQ_largewrong}}\\
\subfloat[Right way]{\includegraphics[width=\columnwidth]{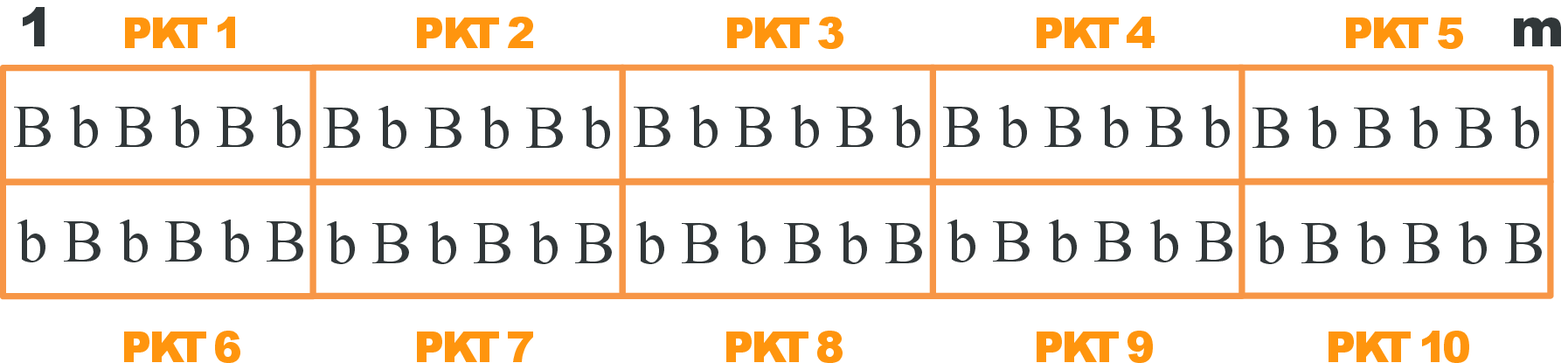}\label{GQ_largescale}}
\caption{Segmenting a long measurements vector.}
\end{figure}

It may happen that the MDC system needs to deal with a large scale problem in which the number of measurements to be acquired is fairly large. In this scenario, the transmitter cannot create packets that are arbitrarly large as their size is capped by the MTU (Maximum Transmission Unit) size of the adopted communication protocol. As an example, the IEEE 802.15.4 protocol, popular in sensor network applications, specifies an MTU equal to 127 bytes with an effective maximum payload size of 104 bytes.

As far as CS-GQ is concerned, a careful segmentation operation must be performed whenever the size of a description exceeds the MTU. The goal is to keep the packets balanced, in order to always get useful information whenever a packet is received. Fig. \ref{GQ_largescale} graphically shows a way to assign quantization step sizes and performing segmentation. Every packet contains an alternating sequence of measurements quantized with $B$ bits and $b$ bits and has a dual packet with the opposite alternation pattern. Therefore, every received packet is informative, in the sense that it is able to improve the current set of measurements, whether by adding new measurements or by improving their accuracy. To fix ideas, Fig. \ref{GQ_largewrong} shows a wrong way to perform segmentation, where the packets containing only low-resolution measurements are not informative if the high resolution dual set is already available. In case of loss of a packet, one could potentially lose more high resolution measurements with respect to the case depicted in Fig. \ref{GQ_largescale}.

\section{Numerical experiments and applications}
\label{numerical}
Numerical simulations have been performed to evaluate the performance of the proposed technique against other methods to implement MDC for CS and to show the advantages for some practical usage scenarios, inspired by sensor network applications. All the results requiring side decoding have been obtained using the side decoder presented in Sec. \ref{admm}. We remark that the algorithm solves problem \eqref{sideGQ} and that the solution returned by the proposed algorithm is the same as the one returned by a solver of convex problems such as CVX. However, the proposed algorithm is specialized for the particular problem considered in this paper and can be significantly faster than CVX. The SeDuMi solver employed by CVX uses second-order methods, such as Newton's method, and the complexity per iteration grows as $\mathcal{O}(n^6)$. ADMM on the other hand is a first order method, thus much less expensive per iteration than a second order method, although in principle having slower convergence. Nevertheless, ADMM enjoys linear convergence rates under mild conditions on the cost function, as shown in \cite{hong2012linear}. As a practical example, solving a problem with $n=1024$, $k=150$ non-zero components, $m=600$ measurements obtained by a Gaussian sensing matrix, $B=6$ bits, $b=4$ bits requires 104 seconds using CVX, while only 9.1 seconds are needed by the proposed algorithm based on ADMM, thus achieving a tenfold speed-up.
Moreover, whenever quantization cells are uniform, we add to the reconstruction problem an $\ell_\infty$ constraint in addition to the $\ell_2$ constraint presented in \eqref{eq:CS_recovery_relaxed}. This is to enforce consistent reconstruction as discussed in Sec.~\ref{cs_gq}.

\subsection{Effectiveness of optimization via oracle formula}
\label{optimal_simulations}
We present an experiment that shows the average distortion obtained via optimization of the parameters with the operational method and with the oracle method discussed in Sec.\ref{GQ_opt}. The curves shown in Fig. \ref{optdistortion} are obtained for the same system parameters of Fig. \ref{tradeoff}, \emph{i.e.}, $m=120$ measurements for a $n=256$-long, $k=10$-sparse signal and a memoryless channel with packet loss probability $p$. It is observed that the optimization via the oracle method provides a good estimate of the parameters, only yielding slightly higher distortion when the packet loss probability is very small. This is due to the ideality of oracle decoding, which estimates lower distortions than actually achieved by practical algorithms such as $l_1$ minimization.

\begin{figure}
\centerline{\includegraphics[width=.9\columnwidth]{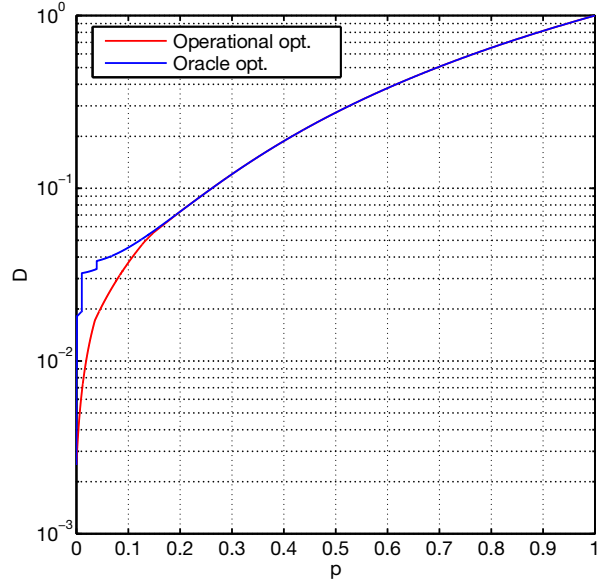}}
\caption{Average distortion over a memoryless lossy channel with packet loss probability $p$. Parameters $b$ and $B$ are optimized via oracle/operational formulas. $n=256$, $m=120$, $k=10$.}
\label{optdistortion}
\end{figure}

\subsection{Comparison with other quantization-based MDC schemes}
In this section we compare the performance of CS-GQ against other methods generating multiple descriptions using quantization. Thus, we compare CS-GQ, CS-SPLIT and CS-MDSQ, where CS-MDSQ is an application of the multiple description scalar quantizer (MDSQ) as introduced in \cite{mdsq} to the quantization of compressive measurements. Although MDSQ is a powerful and popular solution for MDC, and provides asymptotic performance close to the rate-distortion bound, its use in CS is not straightforward. In fact, CS reconstruction is a nonlinear process, with no guarantee that lower distortion on the measurements results in better reconstruction quality. Hence, the choice of the number of measurements and quantization step size is a trade-off in order to achieve the best performance after recovery. In our simulations, we use the index assignment technique developed in \cite{mdsq2011}. As a first comparison we fix the number of measurements and obtain the tradeoff points on the side/central distortion plane for both CS-GQ and CS-MDSQ. In Fig. \ref{tradeoff} it is observed that fixing $m=120$ for a $n=256$-long, $k=10$-sparse signal allows CS-GQ to provide better tradeoff points than CS-MDSQ. However, CS-MDSQ outperfroms CS-GQ when a lower number of measurements is chosen. We remark that, theoretically, CS-MDSQ allows $2^{R+1}$ tradeoff points by modulating the redundancy, while CS-GQ admits $\left\lfloor \frac{R}{2} \right\rfloor +1$ points only. However, observing Fig.~\ref{tradeoff} we can notice that the redundancy of many tradeoff points of CS-MDSQ is so high that the resulting side distortion is severe (\emph{e.g.}, above 0.4). We remark that the MDSQ can also accomodate the zero-redundancy case when all the entries of the index assignment matrix are used, corresponding to CS-SPLIT. Both zero-redundancy CS-MDSQ and CS-SPLIT obtain the same central distortion. However, side distortion may be different. CS-SPLIT represents half of the measurements with a fine quantization rate, while CS-MDSQ represents all the measurements with a coarser quantization rate. The effect on CS recovery is that when there are plenty of measurements (depending on $n$ and on the sparsity of the signal), the side distortion after recovery is quantization-limited, so even though CS-MDSQ provides more measurements, their coarser quantization limits the performance. We also notice that the full-redundancy case ($B=b$) corresponds to a CS-MDSQ with staggered side quantizers and equal size intervals in the central partition.

A further experiment evaluates the dependency of the side and central distortion on the number of measurements. Hence, we fix the number of measurements and an operating point and we analyse how central and side distortion change as a function of the number of measurements. Fig. \ref{gq_vs_theworld} is obtained with the same signal as in Fig. \ref{tradeoff} and choosing $m=120$, $(B,b)=(6,2)$ as the operating point of CS-GQ and $M=2$ as the redundancy of the MDSQ with $R=8$ total rate (namely, the ratio between the step size of the side quantizer and of the central quantizer is $2M$; refer to \cite{mdsq2011} for further details). We observe that the distortion of CS-GQ decreases faster than the distortion of CS-MDSQ as the number of measurement increases, while central distortion decreases at the same rate although CS-GQ is marginally better. The reason behind this performance is that CS-MDSQ is more quantization-limited than CS-GQ when the number of measurements is high. Viceversa, for a very low number of measurements, CS-GQ is measurement-limited.
The previous experiment considered a rather generous overall budget equal to $Rm=960$ bits. Indeed, we considered a fixed number of measurements and chose the rate according to the budget. It must be noticed that one could optimize both the value of $m$ and the value of $R$ under the overall budget constraint and this could lead to different choices for CS-GQ and CS-MDSQ. Although we do not report the results for brevity, we observed that if the budget is very large as in the previous case, CS-MDSQ generally provides better tradeoff curves than CS-GQ by using fewer measurements and a finer quantization rate. However, results change under a tight budget constraint. A minimum number of measurements has to be acquired in order to ensure successful reconstruction and due to the tight budget constraint it is not possible to use high quantization rates. Fig.\ref{tight_budget} compares the tradeoff curves of CS-GQ and CS-MDSQ for the two operating points $m=50$ and $R=8$ and for $m=100$, $R=4$. Those two choices of $m$ and $R$ are optimal for CS-MDSQ and CS-GQ respectively under a budget of $Rm=400$ bits. We notice that CS-GQ outperforms CS-MDSQ despite choosing the best combination of $m$ and $R$. We therefore conclude that CS-GQ can be advantageous with respect to CS-MDSQ when the bit budget is low. 

\begin{figure}
\centering
\subfloat[Side Distortion]{\includegraphics[width=0.94\columnwidth]{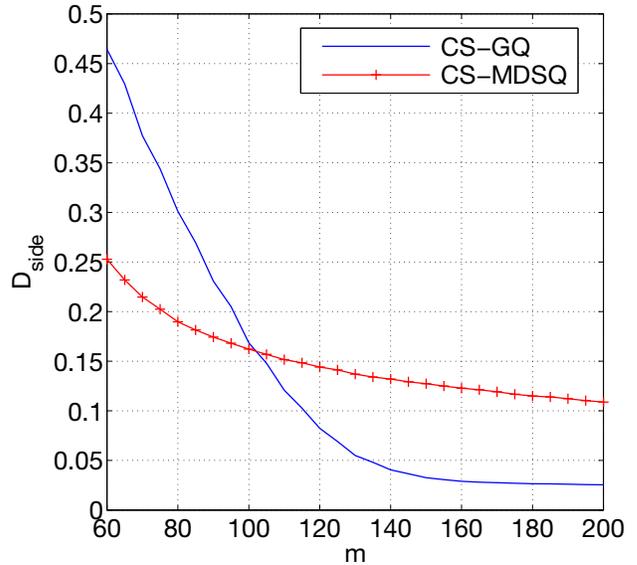}\label{gq_vs_theworld_side}}\\
\subfloat[Central Distortion]{\includegraphics[width=0.94\columnwidth]{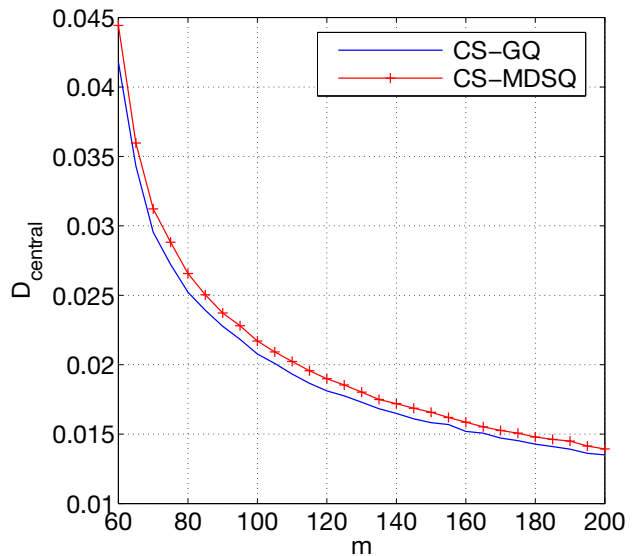}\label{gq_vs_theworld_central}}
\caption{Comparison of CS-GQ and CS-MDSQ. $(B,b)=(6,2)$ for CS-GQ and $M=2,\ R=8$ for CS-MDSQ. $n=256$, $k=10$.}
\label{gq_vs_theworld}
\end{figure}

\begin{figure}
\centering
\includegraphics[width=0.44 \textwidth]{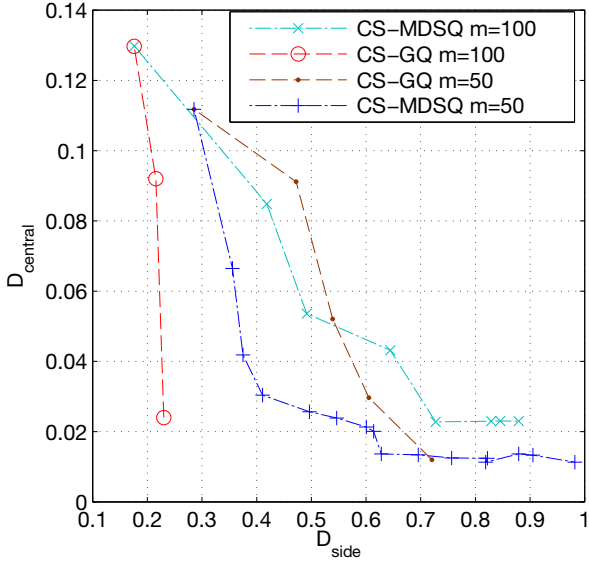}
\caption{Comparison under a tight budget $Rm=400$ bits. Optimal choice for CS-MDSQ is $m=50$,$R=8$. Optimal choice for CS-GQ is $m=100$,$R=4$.}
\label{tight_budget}
\end{figure}

\subsection{Simulations over MTU-limited memoryless and Gilbert channels}
In this section we perform some simulations to assess the gain achieved by using CS-GQ when the communication channel is prone to packet loss. We consider two channel models, which are significant for the performance assessment of CS-GQ in a packetization scenario: the memoryless channel, in which the loss trace is a Bernoulli process, and the Gilbert channel \cite{gilbert,Elliott1963}, where memory is modelled as a two-state Markov chain, as shown in Fig.~\ref{gilbert_markov}. In this model the communication link can be in either of two states, that we label Good (G) and Bad (B) with a probability $p$ of moving from G to B and probability $q$ of moving from B to G. When in B state the link will drop any packet transmitted. This is a popular model for channels exhibiting burst errors, \emph{e.g.}, due to fading in a wireless system. In both cases we assume that a limit to the maximum packet size is present and it is equal to 100 bytes, as in the MTU of IEEE 802.15.4.

\begin{figure}
\centerline{\includegraphics[width=.9\columnwidth]{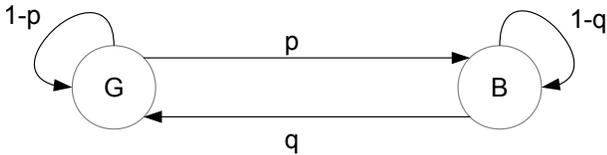}}
\caption{The Gilbert-Elliot channel model. $p$ is the transition probability from good to bad state.}
\label{gilbert_markov}
\end{figure}

The simulations over the memoryless channel have been performed with a signal of length $n=256$, $k=40$ non-zero components, $m=160$ measurements and a total rate $R=10$ bits per measurement over the two descriptions. CS-GQ has been implemented to create two packets of 100 bytes each, in the same fashion of Fig.~\ref{GQ_largescale}, and compared against a standard segmentation of the measurements, \emph{i.e.}, splitting half of the measurements into the two packets of 100 bytes. The tests measured the distortion in the reconstructed signal averaged over $10^5$ runs. The values of $B$ and $b$ are automatically optimized using the oracle method (see Sec. \ref{GQ_opt}). Fig.~\ref{iid} shows that CS-GQ has superior performance, yielding a lower average distortion than segmentation of the measurements and the gain is more significant when the packet loss probability is high.

\begin{figure}
\centerline{\includegraphics[width=.9\columnwidth]{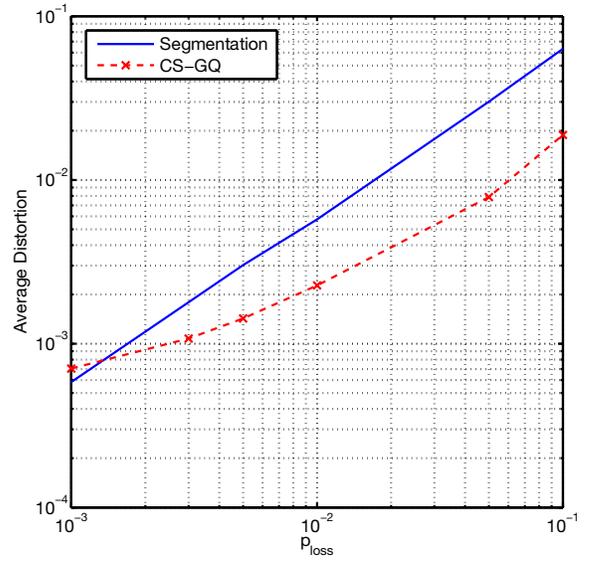}}
\caption{Simulations over memoryless channel.}
\label{iid}
\end{figure}

The simulations over the Gilbert channel have been performed using a longer signal, in order to correctly test the effect of the memory of the channel. We suppose that a batch of $N_s=1000$ vectors of measurements have to be transmitted in sequence. The length of the signals is $n=1000$, $k=200$ non-zero components, $m=720$ measurements and a total rate $R=10$ bits per measurement. Table \ref{gilbert_table} reports the average distortion for some values of the transition probabilities $p$ and $q$. It can be noticed that the proposed scheme allows to achieve lower distortion with respect to the segmentation approach.

\begin{table}
\caption{Average distortion over the Gilbert-Elliot channel}
\centerline{
\begin{tabular}{ c c c c }
\hline 
$p$ & $q$ & \bf{Segmentation} & \bf{CS-GQ}\tabularnewline
\hline 
\hline 
0.05 & 0.5 & 0.0257 & \bf{0.0192}\tabularnewline
\hline 
0.05 & 0.3 & 0.0762 & \bf{0.0496}\tabularnewline
\hline 
0.05 & 0.15 & 0.1820 & \bf{0.1205}\tabularnewline
\hline 
0.01 & 0.3 & 0.0160 & \bf{0.0102}\tabularnewline
\hline 
0.01 & 0.15 & 0.0451 & \bf{0.0324}\tabularnewline
\hline 
\end{tabular}
}
\label{gilbert_table}
\end{table}

\subsection{Simulations of object recognition over lossy channel}
In this section we consider an application recently proposed in \cite{YangDistributedObjectRec} and \cite{YangMultiviewObjectRec} as a possible scenario in which the techniques proposed in this paper could provide a significant performance improvement. The purpose is to show the effectiveness of CS-GQ in a practical scenario. In this application, a smart camera platform, such as CITRIC \cite{CITRIC}, computes histograms of image descriptors to be used as features in a scene classification task. Image descriptors are compact representations of an image that provide a certain degree of invariance to transformations such as rotation, scaling, etc., and are widely used in the computer vision field for visual search or scene recognition tasks. The most popular descriptors are the Scale-Invariant Feature Transform (SIFT) \cite{loweSIFT}, which describe each image keypoint, \emph{i.e.}, a signficant point of the image to be included in the descriptor, through a vector of $s=128$ entries. Each image $i$ has a variable number $N_i$ of keypoints, and thus associated descriptors. A clustering algorithm, such as $\kappa$-means, is used to identify $\kappa$ clusters in $\mathbb{R}^s$ from the $N_i$ original descriptors. A histogram of descriptors is obtained by counting how many keypoint descriptors fall in each cluster. In \cite{YangDistributedObjectRec}, the authors leverage the sparsity of such histograms to compress them by means of random projections, and transmit them to a remote fusion center. The authors do not consider the problems of quantization of the measurements, nor the possibility of having channel losses. However, we remark that those are key problems for this application. In fact, transmitting the random projections as floating point values requires very large bandwidth, hence quantization could reduce such requirement and save a sizeable amount of the scarce energy of the mobile platform. Also, channel losses may occur and retransmission of packets may not be feasible due to real-time constraints, and it would indeed require further energy consumption on sensor side. Hence disabling retransmission and providing a scheme that is robust to channel losses could be of interest for such platforms. 

In our tests, we build a $10,000$-dimensional dictionary by running hierarchical $\kappa$-means, with 10 clusters per level in the hierarchy, over SIFT descriptors  extracted from some images in the BMW database \cite{YangDistributedObjectRec}. This database comprises photos of the UC Berkeley campus acquired with the CITRIC platform with multiple views of the same scene from different angles. We then use the dictionary of quantized SIFT features to produce histograms of descriptors. Each image is characterized by a histogram of descriptors, which is sparse because the number of descriptors extracted from each image is much less than the number of entries in the dictionary. \cite{YangDistributedObjectRec}. A support vector machine is trained from histograms of descriptors of the images in the training set, in order to solve a classification task that amounts to recognizing a scene from a number of possible scenes. In order to simulate a real system, we consider  packets with maximum size of 100 bytes (as in the case of the MTU of IEEE 802.15.4), to be transmitted over a memoryless channel. The test metric that we adopt in this case is the accuracy in the scene recognition. Table \ref{tableaccuracy} shows the accuracy in the scene recognition task as function of the packet loss probability $p$. We can see that CS-GQ builds robustness into the system and allows to achieve higher recognition rates, when in presence of packet losses.

\begin{table}
\caption{Accuracy of scene recognition}
\label{tableaccuracy}
\centerline{
\begin{tabular}{ c c c c }
\hline 
$p$ & \bf{Segmentation} & \bf{CS-GQ}\tabularnewline
\hline 
\hline 
0.01 & 70.63\% & \bf{71.25\%}\tabularnewline
\hline 
0.05 & 61.25\% & \bf{63.75\%}\tabularnewline
\hline 
0.10 & 54.38\% & \bf{61.88\%}\tabularnewline
\hline
\end{tabular}
}
\end{table}

\section{Conclusions}
In this paper we considered the problem of robust transmission of CS measurements over unreliable channels. We discussed some strategies based on the Graded Quantization paradigm that enable to increase robustness in a way that directly exploits the democratic property of compressive measurements. Moreover, contrary to many cases in the literature, we considered a realistic packetization scenario and discussed the issues involved in it as well as the performance of the proposed method. Finally, we proposed a novel reconstruction algorithm, based on ADMM, for the specific problem of having measurements obtained from multiple quantizers. Our simulations showed that the proposed technique is a competitive method to implement MDC for CS applications. In particular, we showed that we can improve scene recognition accuracy in an application of a smart camera platform using CS and transmitting over an unreliable channel.  


\end{document}